\documentclass[12pt]{amsart}

\usepackage{amsthm}
\usepackage{amssymb}
\usepackage{amsmath}
\usepackage{mathrsfs}
\numberwithin{equation}{section}
\newcommand{\bea}{\begin{eqnarray}}
\newcommand{\eea}{\end{eqnarray}}
\newcommand{\be}{\begin{eqnarray*}}
\newcommand{\ee}{\end{eqnarray*}}
\newtheorem{theorem}{Theorem}[section]

\begin{document}
\title[Linear Transformations and RIP]{Linear Transformations and\\ Restricted Isometry Property}
\author{Leslie Ying and Yi Ming Zou}
\thanks{L. Ying is with the Department of Electrical Engineering, University of Wisconsin, Milwaukee, WI 53201, USA, email: leiying@uwm.edu}
\thanks{Y. M. Zou is with the Department of Mathematical Sciences, University of Wisconsin, Milwaukee, WI 53201, USA, email: ymzou@uwm.edu}
\thanks{This paper has been submitted to ICASSP 09}
\date{12/16/08 (This version)}

\maketitle

\begin{abstract}
The Restricted Isometry Property (RIP) introduced by Cand\'es and Tao is a fundamental property in compressed sensing theory. It says that if a sampling matrix satisfies the RIP of certain order proportional to the sparsity of the signal, then the original signal can be reconstructed even if the sampling matrix provides a sample vector which is much smaller in size than the original signal. This short note addresses the problem of how a linear transformation will affect the RIP. This problem arises from the consideration of extending the sensing matrix and the use of compressed sensing in different bases. As an application, the result is applied to the redundant dictionary setting in compressed sensing.  
\end{abstract}
\section{Introduction}
In Compressed Sensing (CS), one considers the problem of recovering a vector (discrete signal) $x\in \mathbb{R}^N$ from its 
linear measurements $y$ of the form 
\bea\label{linmeas} y_i =<x,\varphi_i>,\;1\le i\le n, \eea
with $n<<N$.  If $x$ is sparse, CS theory says that one can actually recover $x$ from $y$ which is much smaller 
in size than $x$ by solving a convex program with a suitably chosen set of sampling row vectors
$\{\varphi_i| 1\le i\le n\}$ \cite{Cand1}\cite{Donoho}\cite{Cand3}. The linear system (\ref{linmeas}) can be written in the form of matrix multiplication 
\bea y=\Phi x, \eea
where $\Phi$ is an $n\times N$ matrix formed by the row vectors $\varphi_i$ called the sampling matrix. One of the Conditions that ensures the performance of the sampling matrix $\Phi$ is the RIP.  
A matrix $\Phi\in\mathbb{R}^{n\times N}$ is said to satisfy the RIP of order
$k\in\mathbb{N}$ and {\it isometry constant} $\delta_k\in(0,1)$ if
\bea\label{rip}
(1-\delta_k)\|z\|_2^2\le\|\Phi_T z\|_2^2\le(1+\delta_k)\|z\|_2^2,
\qquad \forall z\in\mathbb{R}^{|T|},
\eea
where $T\subset\{1,2,\ldots,N\}$ satisfying $|T|\le k$, and $\Phi_T$ denotes
the matrix obtained by retaining only the columns of $\Phi$
corresponding to the entries of $T$. Condition (\ref{rip}) is equivalent to the condition that all the matrices $\Phi_{T}^{t}\Phi_T$ have their eigenvalues in $[1-\delta_k, 1+\delta_k]$. For any matrix $X\in \mathbb{R}^{r\times s}$ and any $k\in\mathbb{N}$, we denote the corresponding isometry constant of $X$ by $\delta_k(X)$. If there is no confusion, we will just write $\delta_k$. In particular, we always use $\delta_k$ for the matrix $\Phi$.
\par
A theorem due to Cand\'es, Romberg, and Tao \cite{Cand2} says that if
$\Phi$ satisfies the RIP of order $3k$, then the solution $\hat{x}$ of the following convex minimization
problem
\bea
\text{min} \|x\|_{1} \ \ \ \ \ \text{ subject to }\  \Phi x=y,
\eea
satisfies (see also \cite{Baraniuk})
\bea
\|x-\hat{x}\|_{2}\le \frac{C_2\sigma_k(x)}{\sqrt{k}},
\eea
where $\sigma_k(x)$ is the $\ell_1$ error of the best $k$-term approximation, and $C_2$ is a constant depending only
on $\delta_{3k}\in (0,1)$.\footnote{It should be noted that the RIP is only a sufficient condition for reconstruction. If $\Phi$ satisfying the RIP, $cA$ may not satisfy the RIP for $c\ne 0$. However, it is clear that both $A$ and $cA$ lead to similar sparse recovery using $\ell_1$ program. However, this issue is beyond the current scope \cite{DG}.}
\par
A condition that ensures a random matrix satisfies the RIP with high probability is given by the concentration of measure 
inequality. An $n\times N$ random matrix $\Phi$ is said to satisfy the concentration of measure inequality if for any
$x\in \mathbb{R}^N$,
\bea\label{cmi}
P(|\|\Phi x\|^2_2 - \|x\|_2^2|\ge \varepsilon \|x\|_2^2)\le 2e^{-nc_0(\varepsilon)},
\eea
where $\varepsilon\in (0,1)$, and $c_0(\varepsilon)$ is a constant depending only on $\varepsilon$.
\par
The random matrices $\Phi=(r_{ij})$ generated by the following probability distributions are known to satisfy the concentration of measure inequality with $c_0(\varepsilon)=\varepsilon^2/4-\varepsilon^3/6$ \cite{Baraniuk}:
\begin{equation*}
r_{ij}\sim N\left(0,\frac{1}{n}\right),
\end{equation*}
\begin{equation}\label{eqn:distr}
r_{ij}=\left\{
\begin{array}{ccc}
\displaystyle\frac{1}{\sqrt{n}}&\text{with probability} &1/2\\
-\displaystyle\frac{1}{\sqrt{n}}&\text{with probability} &1/2\\
\end{array}
\right..
\end{equation}
\par
According to Theorem 5.2 in \cite{Baraniuk}\footnote{In the proof given in \cite{Baraniuk}, the constant $c_1$ was first chosen such that $a:=c_0(\varepsilon)\delta/2-c_1[1+(1+\log\frac{12}{\delta})/\log\frac{N}{k}]>0$, then the constant $c_2$ was chosen such that $0<c_2<a$. Thus the constants depend also on $\varepsilon$.}, for given integers $n$ and $N$, and $0<\delta<1$, 
if the probability distribution generating the $n\times N$ matrices $\Phi$ satisfies 
the concentration inequality (\ref{cmi}), then there exist constants $c_1, c_2 > 0$ depending only on $\delta$ such 
that the RIP holds for $\Phi$ with the prescribed $\delta$ and any 
\bea\label{krange} k\le c_1n/\log(N/k) \eea
with probability $\ge 1-e^{-2c_2n}$. Furthermore, this RIP for $\Phi$ is universal in sense that it holds with respect to any orthogonal basis used in the measurement. 
\par
There are also deterministic constructions of matrices satisfying the RIP \cite{DeVore}\cite{VS}\cite{HS}\cite{GHS}. 
\par
For application purposes, one often needs to analyze the RIP constants of the products of a matrix $\Phi$ with known RIP constant $\delta$ and other matrices. For example, when one considers different bases or redundant dictionaries 
under which the signals of interest are sparse, matrices of the form $\Phi B$ needs to be analyzed \cite{Donoho}\cite{Rauhut}, where $B$ is given by the basis or the dictionary. For another example, if the size of $\Phi$ is $n\times N$ with $n < N$, one would like to extend $\Phi$ to $A\Phi B$ of size $m\times q$ with $m<n<N<q$ if possible, since that gives a further reduction on the number of measurements one needs to collect: for $\Phi$, the number of measurements is $n$; while for $A\Phi B$, the number of measurements is $m$. 
\par
These situations can be formulated under a more general framework by asking the following question: If a matrix $\Phi$ of size $n\times N$ satisfies the RIP with a given isometry constant $0<\delta<1$ (with certain probability if $\Phi$ is random), and $A,B$ are given matrices of sizes $m\times n$ and $N\times q$ respectively, then what is the isometry constant of the matrix $A\Phi B$?
\par
 In section 2, we first show that if all $\Phi$, $A$, and $B$ are random and satisfy the concentration of measure inequality, then $A\Phi B$ satisfies the concentration of measure inequality, therefore it has RIP. Then we observe that if deterministic matrix is involved, the problem is more complicated, but it can still be analyzed by using the SVDs of $A$ and $B$. It is not possible to multiply by a deterministic $A$ from the left to achieve more reduction on the number of measurements without further assumption. Our result shows that it is possible to extend the matrix $\Phi$ by multiplying a deterministic $B$ from the right to extend $\Phi$ if $\Phi$ is random, though the isometry constant will be changed. This result can be applied to redundant dictionary setting to give a different approach for using CS with redundant dictionaries.
\par
\section{Main result}
\par
We first consider the random case. Let $\Phi$ be an $n\times N$ matrix satisfying the concentration inequality (\ref{cmi}) with constant $\varepsilon$, and let $A$ (respectively $B$) be a random matrix size $m\times n$ (respectively $N\times q$) satisfying the concentration inequality (\ref{cmi}) with $\varepsilon_1$ (respectively, $\varepsilon_2$). Then we have:
\begin{theorem}
Assume that all $\varepsilon, \varepsilon_1,\varepsilon_2 < 1/3$. The matrix $A\Phi$ satisfies the concentration inequality
\be
P(|\|A\Phi x\|^2_2 - \|x\|_2^2|\ge \varepsilon_3 \|x\|_2^2)\le 2e^{-mc^{\prime}_0},
\ee
where $\varepsilon_3=\varepsilon+\varepsilon_1(1+\varepsilon)$, and $c^{\prime}_0$ is a constant that depends only on $c_0(\varepsilon)$ and $c_0(\varepsilon_1)$ (as defined in (\ref{cmi})). The same statement holds for $\Phi B$ with $\varepsilon_3=\varepsilon+\varepsilon_2(1+\varepsilon)$ and $m$ replaced by $n$.
\end{theorem}
\begin{proof}
We give the proof for the case of left multiplication by $A$, the proof for the case of right multiplication by $B$ is similar. By assumption, with probability $\ge 1-2e^{-mc_0(\varepsilon_1)}$, the matrix $A$ satisfies 
\be
(1-\varepsilon_1)\|y\|_2^2 <\|A y\|_2^2 <(1+\varepsilon_1)\|y\|_2^2,\quad \mbox{ for any $y\in \mathbb{R}^{n}$}.
\ee
Replacing $y$ by $\Phi x$ ($x\in\mathbb{R}^N$), we have
\bea\label{thm211}
(1-\varepsilon_1)\|\Phi x\|_2^2 <\|A \Phi x\|_2^2 <(1+\varepsilon_1)\|\Phi x\|_2^2.
\eea
Again by assumption, with probability $\ge 1-2e^{-nc_0(\varepsilon)}$, the matrix $\Phi$ satisfies 
\bea\label{thm212}
(1-\varepsilon)\|x\|_2^2 <\|\Phi x\|_2^2 <(1+\varepsilon)\|x\|_2^2,\quad \mbox{for any $ x\in \mathbb{R}^{N}$}.
\eea
Now the statement follows by combining (\ref{thm211}) and (\ref{thm212}).
\end{proof}
\par\medskip
\noindent\textbf{Remark.} If $m\le n$, the constant $c^{\prime}_0$ in Theorem 2.1 can be roughly estimated by the inequality $c^{\prime}_0\le c_0(\varepsilon^{\prime})-\log 2/m $, where $c_0(\varepsilon^{\prime})= min\{c_0(\varepsilon_1), c_0(\varepsilon)\}$. This is obtained from 
\be
1-(1-2e^{-mc_0(\varepsilon_1)})(1-2e^{-nc_0(\varepsilon)})\le 2e^{-m(c_0(\varepsilon^{\prime})-\log 2/m)}.
\ee
 More precise estimation can be carried out, but we are not concerning this point here.
\par
\medskip
Now we consider the cases when deterministic matrices are involved. We observe that it is not possible to multiply a deterministic matrix $A$ from the left to extend the sensing matrix to achieve further reduction in sampling without other assumptions. To see this, we consider the SVD of $A$.
\par
For any positive integer $d$, let $O(d)$ be the set of $d\times d$ orthogonal matrices. There exists $U\in O(n)$ such that 
\bea
A^tA=U^t\left(\begin{array}{cccc} 
\sigma_1 & {} & {} & {}\\
{} & \sigma_{2} & {} & {}\\
{} & {} & \ddots & {}\\
{} & {} & {} & \sigma_{n} \end{array}\right)U,
\eea
where $\sigma_1\ge\sigma_2\ge\cdots\ge\sigma_n\ge 0$. Since for any $T\subset\{1,2,\ldots,N\}$, $(A\Phi)_T=A\Phi_T$, we have
\bea\label{svd1}
(A\Phi)_T^t(A\Phi)_T 
&=& \Phi_T^tA^tA\Phi_T\\\nonumber
&=& \Phi_T^tU^t\left(\begin{array}{cccc} 
\sigma_1 & {} & {} & {}\\
{} & \sigma_{2} & {} & {}\\
{} & {} & \ddots & {}\\
{} & {} & {} & \sigma_{n} \end{array}\right)U\Phi_T\\\nonumber
&=& (U\Phi)_T^t\left(\begin{array}{cccc} 
\sigma_1 & {} & {} & {}\\
{} & \sigma_{2} & {} & {}\\
{} & {} & \ddots & {}\\
{} & {} & {} & \sigma_{n} \end{array}\right)(U\Phi)_T.
\eea
If $m<n$, then $\sigma_{m+1}=\cdots=\sigma_n=0$, and hence 
\be
\left(\begin{array}{cccc} 
\sigma_1 & {} & {} & {}\\
{} & \sigma_{2} & {} & {}\\
{} & {} & \ddots & {}\\
{} & {} & {} & \sigma_{n} \end{array}\right)(U\Phi)_T=\left(\begin{array}{c} A_1 \\ 0 \end{array}\right)
\ee
for a suitable block matrix $A_1$. From the last matrix one can see immediately that RIP fails.
\par
 If $m\ge n$, then we can change $\Phi$ by multiplying $A$ from the left if $A$ has full column rank. Since under this assumption, all $\sigma_i>0$. Note that $U\Phi$ has the same RIP as $\Phi$, so if $\delta_k$ is the RIP constant of $\Phi$ corresponding to all $T$ of size $k\le N$, we can bound the RIP constant of $A\Phi$ by $\sigma_n(1-\delta_k)$ and $\sigma_1(1+\delta_k)$. In fact, for $z\in\mathbb{R}^k$, if we let $U\Phi_Tz=y=(y_1,\ldots,y_n)^t$, then $\|y\|_2=\|\Phi_T z\|_2$, and according to (\ref{svd1})
\bea
\sigma_n\|y\|_2^2\le\|A\Phi_T z\|_2^2 = \sum_{i=1}^n\sigma_iy_i^2\le\sigma_1\|y\|_2^2.
\eea
Thus we have (use (\ref{rip}))
\bea\label{rip1}
\sigma_n(1-\delta_k)\|z\|_2^2\le\|A\Phi_T z\|_2^2\le\sigma_1(1+\delta_k)\|z\|_2^2.
\eea
Note that the above analysis works whether $\Phi$ is random or deterministic. 
\par\medskip
\medskip
Next, we consider the product $\Phi B$. In this case, we need to distinguish between random matrix $\Phi$ and deterministic matrix $\Phi$. Assume that $\Phi$ is a random matrix satisfying the concentration inequality (\ref{cmi}) and hence satisfying the RIP inequality (\ref{rip}) with probability $\ge p$. Note that the concentration inequality is invariant under the right multiplication by an orthogonal matrix. That is, if $U\in O(N)$, then $\Phi U$ also satisfies (\ref{rip}) with probability $\ge p$.
\par
Let $B$ be an $N\times q$ matrix. To make the argument clearer, we assume that $T\subset\{1,2,\ldots,q\}$ with $|T|= k<N$ (note that this is sufficient for our purpose). We have $U\in O(N)$ and $V\in O(k)$ such that
\bea
B_T=U\left(\begin{array}{l} D\\ 0\end{array}\right)_{N\times k}V,
\eea
where
\be
D=\left(\begin{array}{cccc} 
\lambda_1 & {} & {} & {}\\
{} & \lambda_{2} & {} & {}\\
{} & {} & \ddots & {}\\
{} & {} & {} & \lambda_{k} \end{array}\right),\quad \lambda_1\ge\lambda_2\ge\cdots\ge\lambda_k\ge 0.
\ee
\par
For $x\in \mathbb{R}^k$,
\bea
\Phi B_Tx=\Phi U\left(\begin{array}{l} D\\ 0\end{array}\right)_{N\times k}Vx.
\eea
Let 
\be
z=\left(\begin{array}{l} D\\ 0\end{array}\right)Vx\in \mathbb{R}^N.
\ee
Then $z$ is $k$-sparse (the last $N-k$ entries are always $0$). Thus, since $\Phi U$ has the same RIP as $\Phi$, we have
\bea\label{rip2}
(1-\delta_k)\|z\|_2^2\le\|\Phi B_Tx\|_2^2=\|\Phi Uz\|_2^2\le(1+\delta_k)\|z\|_2^2
\eea
with probability $\ge p$. 
\par
Let $y=(y_1,\ldots, y_k)^t=Vx$, then $\|y\|_2=\|x\|_2$, and 
\be
\|z\|_2^2 &=& x^tV^t(D^t\;\; 0)\left(\begin{array}{c} D\\ 0\end{array}\right)Vx\\
        {}&=& y^t\left(\begin{array}{cccc} 
\lambda_1^2 & {} & {} & {}\\
{} & \lambda_{2}^2 & {} & {}\\
{} & {} & \ddots & {}\\
{} & {} & {} & \lambda_{k}^2 \end{array}\right)y
        = \sum_{i=1}^k\lambda_i^2y_i^2.
\ee
Since
\be
\lambda_k^2\|y\|_2^2\le \sum_{i=1}^k\lambda_i^2y_i^2\le\lambda_1^2\|y\|_2^2,
\ee
by (\ref{rip2}), we have
\bea\label{rip3}
\lambda_k^2(1-\delta_k)\|x\|_2^2\le\|\Phi B_Tx\|_2^2\le\lambda_1^2(1+\delta_k)\|x\|_2^2
\eea
with probability $\ge p$. 
\par
If $\Phi$ is deterministic, then for arbitrary $U\in O(N)$, $\Phi U$ may not satisfy the same RIP as $\Phi$, and we do not have a good analysis of $\Phi B$ for this case at the moment. Summarize our discussion, we have:
\begin{theorem} 
Notation as before. 
\par
(1) If $A$ is deterministic, then regardless whether $\Phi$ is random or deterministic, $A\Phi$ has RIP if and only if $A$ has full column rank. If that is the case, the RIP constant of $A\Phi$ can be obtained from (\ref{rip1}). If $\Phi$ is random, then the probability for $A\Phi$ to satisfy RIP is the same as that of $\Phi$ (with possible different RIP constant).
\par
(2) If $\Phi$ is a random matrix satisfying the concentration inequality (\ref{cmi}) (hence satisfying the RIP (\ref{rip}) with probability at least $p$), and $B$ is an $N\times q$ deterministic matrix such that $\delta_k(B)\in(0,\frac{2}{1+\delta_k})$, then with probability at least 
\bea\label{thm22}
1-\left(\begin{array}{c} q\\ k \end{array}\right)(1-p),
\eea
the matrix $\Phi B$ satisfies the RIP with the same order as that of $\Phi$ and a possible different RIP constant $\delta_k(\Phi B)$ determined by (\ref{rip3}). 
\end{theorem}
\par
\section{Redundant Bases in Compressed Sensing}
\par
In this section, we apply Theorem 2.2 to redundant bases setting in compressed sensing. From (\ref{krange}), we see that for given $N$ and $k$, the random matrices of size $n\times N$ generated by the distributions described in (\ref{eqn:distr}) satisfy the RIP with high probability as long as
\be
n\ge Ck\log(N/k)\quad\mbox{for some constant $C$}.
\ee
Therefore it is desirable to reduce the integer $k$, i.e. to increase the sparsity level of the signal, by considering redundant bases (or redundant dictionaries). Recall that if a set of vectors $\mathbf B$ spans a vector space $V$, then we call $\mathbf B$ a basis if $\mathbf B$ is linearly independent and call $\mathbf B$ a redundant basis otherwise. To apply compressed sensing to a signal $y\in \mathbb{R}^N$ that has a sparse representation $x$ under a redundant basis $\mathbf B$ of size $q> N$, we need to consider how the combination of a good sensing matrix with a redundant basis affects the RIP.
\par 
Let $B$ be the matrix corresponds to the redundant basis $\mathbf B$. Then $B$ is of size $N\times q$ and $y=Bx$ with $x\in \mathbb{R}^q$ sparse. This problem has been considered in \cite{Donoho}\cite{Rauhut}. In particular, in \cite{Rauhut}, a detailed analysis of the situation was given. According to Theorem 2.2 in \cite{Rauhut}, if $\Phi$ satisfies the concentration inequality (\ref{cmi}) with \footnote{There should be a factor $S$ (which is our $k$) for the term $\log(e(1+12/\delta))$ in the bound for $n$ given in \cite{Rauhut}. This affects some later estimates in \cite{Rauhut}.}
\bea\label{nrange}
n\ge C\delta_k^{-2}[k(\log(N/k)+\log e(1+12/\delta_k))+\log 2+t],
\eea
for some $\delta_k\in (0,1)$ and $t>0$, then with probability at least $1-e^{-t}$, the restricted isometry constant of $\Phi B$ satisfies
\bea\label{dd}
\delta_k(\Phi B)\le \delta_k(B) +\delta_k(1+\delta_k(B)).
\eea
\par
We now apply Theorem 2.2 to obtain a similar result.
\par
\begin{theorem} Notation as above. With the isometry constant satisfying
\bea
\delta_k(\Phi B)\le \delta_k(B) +\delta_k(\Phi)(1+\delta_k(B))
\eea
and the probability bound given by (\ref{thm22}), the matrix $\Phi B$ satisfies the RIP with the same order as that of $\Phi$.
\end{theorem}
\begin{proof} One just needs to note that the numbers $\lambda_k$ and $\lambda_1$ which appear in (\ref{rip3}) satisfy
\be
1-\delta_k(B)\le \lambda_k^2\le\lambda_1^2\le 1+\delta_k(B).
\ee
\end{proof} 
\par
\medskip
For examples of redundant bases satisfying the condition in Theorem 3.1, we refer the readers to \cite{Rauhut}.
\section{Conclusion and Discussion}
\par
We analyzed the problem of how the multiplication of a matrix to a good sensing matrix affects its RIP. This type of problems arise in CS when one wants to extend the sensing matrix by taking the product of the sensing matrix with another matrix. A particular interesting example is the application of CS under the redundant bases setting. Our result in this short note provides some basic theory for further investigation on the RIP and its applications in CS under different settings. Future work includes constructing good redundant bases, which is related to constructing good deterministic sensing matrices, and analyzing their properties under CS.
\par
\section{Acknowledgments} 
This work is supported in part by the University of Wisconsin System Applied Research Grant and the University of Wisconsin-Milwaukee Research Growth Initiative Grant.
\par

\end{document}